\documentclass[runningheads]{llncs}
 
\usepackage{a4wide}
\usepackage[latin1]{inputenc}
\usepackage{graphicx}
\usepackage{amsfonts}
\usepackage{amssymb}
\usepackage{amsmath}
\usepackage[linesnumbered,boxed,ruled,vlined]{algorithm2e}
\usepackage{framed}
\usepackage{latexsym}
\usepackage{color}
\usepackage{calrsfs}
\usepackage[mathscr]{eucal}
\usepackage{url}
\usepackage{thm-restate}

\def\sC{{\mathscr C}}

\def\sI{{\mathscr I}}

\DeclareMathAlphabet{\mathpzc}{OT1}{pzc}{m}{it}

\def\zD{\mathpzc{D}}
\def\zE{\mathpzc{E}}
\def\zG{\mathpzc{G}}

\def\zM{\mathpzc{M}}
\def\zR{\mathpzc{R}}
\def\zS{\mathcal{S}}

\def\zU{\mathpzc{U}}
\def\zV{\mathpzc{V}}

\newcommand{\OO}[1]{O\left( #1 \right)}
\newcommand{\Th}[1]{\Theta\left( #1 \right)}
\newcommand{\Tht}[1]{\tilde{\Theta}\left( #1 \right)}
\newcommand{\OOt}[1]{\tilde{O}\left( #1 \right)}
\newcommand{\Om}[1]{\Omega \left( #1 \right)}
\newcommand{\Omt}[1]{\tilde{\Omega} \left( #1 \right)}

\newcommand{\CS}{T_s}
\newcommand{\CC}{T_c}
\newcommand{\CU}{T_u}

\newcommand{\Ft}{\mathbb{F}_2}

\newcommand{\eqdef}{\mathop{=}\limits^{\triangle}}

\SetKwRepeat{RepeatU}{repeat}{until}%
\SetKwRepeat{Repeat}{repeat}{}%

\begin{document}
\title{Quantum Information Set Decoding Algorithms}
\author{Ghazal Kachigar \inst{1}  \and Jean-Pierre Tillich \inst{2}}

\institute{Institut de Math\'ematiques de Bordeaux\\ 
Universit\'e de Bordeaux \\
Talence Cedex F-33405, France\\
\email{ghazal.kachigar@u-bordeaux.fr}
\and
Inria,   EPI SECRET\\
2 rue Simone Iff, Paris 75012, France\\
\email{jean-pierre.tillich@inria.fr}
}

\maketitle

\begin{abstract}
The security of code-based cryptosystems such as the McEliece cryptosystem relies primarily on the difficulty of decoding random linear codes. 
The best decoding algorithms are all improvements of an old algorithm due to Prange:  they are known under the name of  information set decoding techniques.
It is also important to assess the security of such cryptosystems against a quantum computer. This research thread started in \cite{OS09} and the
best algorithm to date has been Bernstein's quantising \cite{B10} of the simplest information set decoding algorithm, namely Prange's algorithm.
It consists in applying Grover's quantum search  to obtain a quadratic speed-up of Prange's algorithm.
In this paper, we quantise other information set decoding 
algorithms by using quantum walk techniques which were devised for the subset-sum problem in \cite{BJLM13}.
This results in improving the worst-case complexity of $2^{0.06035n}$ of Bernstein's algorithm to
$2^{0.05869n}$ with the best algorithm presented here (where $n$ is the codelength).
\end{abstract}

\textbf{Keywords:} code-based cryptography, quantum cryptanalysis, decoding algorithm.
\section{Introduction}

As humanity's technological prowess improves, quantum computers have moved from the realm of theoretical constructs to that of objects whose consequences for our other technologies, such as cryptography, must be taken into account. Indeed, currently prevalent public-key cryptosystems such as RSA and ECDH are vulnerable to Shor's algorithm \cite{S97}, which solves factorisation and the discrete logarithm problem in polynomial time. Thus, in order to find a suitable replacement, it has become necessary to study the impact of quantum computers on other candidate cryptosystems. Code-based cryptosystems such as the McEliece \cite{M78} and the Niederreiter \cite{N86} cryptosystems are such possible candidates. 

Their security essentially relies on decoding a linear code. 
Recall that the decoding problem consists, when given a linear code $\sC$ and a noisy codeword $c+e$, in recovering $c$, where $c$ is an unknown codeword of 
$\sC$ and $e$ an unknown error of Hamming weight $w$. A (binary) linear code $\sC$ of dimension $k$ and length $n$ is specified by 
a full rank binary matrix $H$ (i.e. a parity-check matrix) of size $(n-k)\times n$ as 
$$\sC = \{c \in \Ft^n:Hc^T = 0\}.$$
 Since $H(c+e)^T = Hc^T + He^T=He^T$ the decoding problem
can be rephrased as a syndome decoding problem

\begin{problem}[Syndrome Decoding Problem]
Given $H$ and $s^T = He^T$, where $|e| = w$, find $e$.
\end{problem}
This problem has been studied since the Sixties and despite significant efforts on this issue \cite{P62,S88,D91,B97b,MMT11,BLP11,BJMM12,MO15}
the best algorithms for solving this problem \cite{BJMM12,MO15} are exponential in the number of errors that have to be corrected:
correcting $w$ errors in a binary linear code of length $n$ and dimension $k$  has with the aforementioned algorithms a cost  of
$\tilde O(2^{\alpha(\frac{k}{n}, \frac{w}{n})n})$ where $\alpha(R,\omega)$ is positive when $R$ and $\omega$ 
 are both positive.
All these algorithms use in a crucial way the original idea due to Prange \cite{P62} and are known under the name
of Information Set Decoding (ISD) algorithms: they all take advantage of the fact that there might exist a rather large set of positions
containing an information set of the code\footnote{An information set of a linear code $C$ of dimension $k$ is a set $\sI$ of $k$ positions such that 
when given $\{c_i:i \in \sI\}$ the codeword $c$  of $C$ is specified entirely.}
that is almost error free.

All the efforts that have been spent on this problem have only managed to decrease slightly this exponent $\alpha(R,\omega)$.
The following table gives an overview of the average time complexity of currently existing classical algorithms 
when $w$ is the Gilbert-Varshamov distance $d_{\text{GV}}(n,k)$ of the code. 
This quantity is defined by $d_{\text{GV}}(n,k) \eqdef n H_2^{-1}\left(1 - \frac{k}{n}\right)$ where $H_2$ is the binary entropy function $H_2(x) \eqdef -x \log_2(x)-(1-x)\log_2(1-x)$ and $H_2^{-1}$ its inverse
defined from $[0,1]$ to $[0,\frac{1}{2}]$. It corresponds to the largest distance for which we may still expect a unique solution to the decoding problem. If we want uniqueness of the solution, it can therefore be 
considered as the hardest instance of decoding.
In the following table, $\omega_{\text{GV}}$ is defined by the ratio $\omega_{\text{GV}} \eqdef d_{\text{GV}}(n,k)/n$.
\begin{center}
      \begin{tabular}{| c | c | c |}
      \hline
      \textbf{Author(s)} & Year & $\underset{0 \leq R \leq 1}{\max}\alpha(R,\omega_{\text{GV}})$ to 4 dec. places\\
      \hline
      Prange \cite{P62} & 1962 & 0.1207 \\
      \hline
      Dumer \cite{D91} & 1991 & 0.1164 \\
      \hline
      MMT \cite{MMT11} & 2011 & 0.1114 \\
      \hline
      BJMM \cite{BJMM12} & 2012 & 0.1019\\
      \hline
      MO \cite{MO15}& 2015 &  0.0966\\
      \hline
      \end{tabular}\\
      \end{center}
The question of using quantum algorithms to speed up ISD decoding algorithms was first put forward  in \cite{OS09}.
However, the way Grover's algorithm was used in \cite[Subsec. 3.5]{OS09} to speed up decoding did not allow for significant improvements over classical ISD algorithms.
Later on, it was shown  by Bernstein in \cite{B10} that it is possible to obtain much better speedups with Grover's algorithm: by using it for finding an error-free information set, the exponent of Prange's algorithm can indeed be halved. 

This paper builds upon this way of using Grover's search algorithm, as well as the quantum algorithms developped by Bernstein, Jeffery, Lange and Meurer in \cite{BJLM13} to solve the subset sum problem more efficiently.
The following table summarises the ingredients and average time complexity of the  algorithm of \cite{B10} and the new quantum algorithms presented in this paper.
\begin{center}
      \begin{tabular}{| c | c | c | c |}
      \hline
      \textbf{Author(s)} & Year & Ingredients & $\underset{0 \leq R \leq 1}{\max}\alpha(R,\omega_{\text{GV}})$ to 5 dec. places\\
      \hline
      Bernstein \cite{B10} & 2010 & Prange+Grover & 0.06035 \\
      \hline
      This paper & 2017 & Shamir-Schroeppel+Grover+QuantumWalk & 0.05970 \\
      \hline
      This paper & 2017 & MMT+``1+1=0''+Grover+QuantumWalk & 0.05869 \\
      \hline
      \end{tabular}\\
  \end{center}
A quick calculation shows that the complexity exponent of our best quantum algorithm, $MMTQW$, fulfils $\alpha_{\text{MMTQW}} \approx \frac{\alpha_{\text{Dumer}}}{2} + 4.9 \times 10^{-4}$. Thus, our best quantum algorithm improves in a small but non-trivial way on \cite{B10}. Several reasons will be given throughout this paper on why it has been difficult to do better than this.

\noindent
{\bf Notation.} Throughout the paper, we denote by $|e|$ the Hamming weight of a vector $e$. We use the same notation for denoting the cardinality of a set, i.e. $|\zS|$ denotes the cardinality of the set $\zS$. 
The meaning of this notation  will be clear from the context
and we will use calligraphic letters to denote sets: $\zS,\sI,\zM,\dots$. 
We use the standard $\OO{}$, $\Om{}$, $\Th{}$ notation and use the less standard 
$\OOt{}$, $\Omt{}$, $\Tht{}$ notation to mean ``$\OO{}$, $\Om{}$, $\Th{}$, when we ignore logarithmic factors''. 
Here all the quantities we are interested in are functions of the codelength $n$ and we write 
$f(n) = \OOt{g(n)}$ for instance, when there exists a constant $k$ such such that 
$f(n) = \OO{g(n) \log^k(g(n))}$.

\section{Quantum search algorithms}
\subsection{Grover search}
Grover's search algorithm \cite{G96a,G97} is, along with its generalisation \cite{BBHT98} which is used in this paper, an optimal algorithm for solving the following problem with a quadratic speed-up compared to the best-possible classical algorithm.
\begin{problem}[Unstructured search problem]
Given a set $\zE$ and a function $f : \zE \rightarrow \{0,1\}$, find an $x \in \zE$ such that $f(x) = 1$.
\end{problem}
In other words, we need to find an element that fulfils a certain property, and $f$ is an oracle for deciding whether it does. Moreover, in the new results presented in this paper, $f$ will be a quantum algorithm.
If we denote by $\varepsilon$ the proportion of elements $x$ of $\zE$ such that $f(x) = 1$, Grover's algorithm solves the problem above using $O(\frac{1}{\sqrt{\varepsilon}})$ queries to $f$, whereas in the classical setting this cannot be done with less than $O(\frac{1}{\varepsilon})$ queries.
Furthermore, if the algorithm $f$ executes in time $T_f$ on average, the average time complexity of Grover's algorithm will be $O(\frac{T_f}{\sqrt{\varepsilon}})$.
\subsection{Quantum Walk}
\subsubsection{Random Walk.}
Unstructured search problems as well as search problems with slightly more but still minimal structure may be recast as graph search problems.
\begin{problem}[Graph search problem]
Given a graph $G=(\zV,\zE)$ and a set of vertices $\zM \subset \zV$, called the set of \textit{marked elements}, find an $x \in \zM$.
\end{problem}
The graph search problem may then be solved using random walks (discrete-time Markov chains) on the vertices of the graph.
From now on, we will take the graph to be undirected, connected, and $d$-regular, i.e. such that each vertex has exactly $d$ neighbours.
\par{\em Markov chain.}
A Markov chain is given by an initial probability distribution $v$ and a stochastic transition matrix $M$. The transition matrix of a random walk on a graph (as specified above) is obtained from the graph's adjacency matrix $A$ by
$M=\frac{1}{d}A$.
\par{\em Eigenvalues and the spectral gap.}
A closer look at the eigenvalues and the eigenvectors of $M$ is needed in order to analyse the complexity of a random walk on a graph.
The eigenvalues will be noted $\lambda_i$ and the corresponding eigenvectors $v_i$. We will admit the following points (see \cite{CDS80}):\\
(i) all the eigenvalues lie in the interval $[-1,1]$;\\
 (ii) $1$ is always an eigenvalue, the corresponding eigenspace is of dimension $1$;\\
(iii) there is a corresponding eigenvector which is also a probability distribution (namely the uniform distribution $u$ over the vertices). 
It is the unique stationary distribution of the random walk.\\
  We will suppose that the eigenvalues are ordered from largest to smallest, so that $\lambda_1=1$ and $v_1=u$.
An important value associated with the transition matrix of a Markov chain is its \textit{spectral gap}, defined as 
$\delta \eqdef  1 - \max_{i=2,...,d}|\lambda_i|$.
Such a random walk on an undirected regular graph is always {\em reversible} and it is also {\em irreducible} because we have assumed that the graph is 
connected. The random walk is {\em aperiodic} in such a case if and only if the spectral gap $\delta$ is positive. In such a case,  a long enough random walk in the graph converges to the uniform distribution since for all $\eta>0$, we have $||M^kv - u|| < \eta$ for $k=\tilde O(1/\delta)$, where $v$ is the initial probability distribution.

Finding a marked element by running a Markov chain on the graph just consists in 

 \begin{algorithm}[H]
  \DontPrintSemicolon
  \KwIn{$G = (\zE,\zV)$, $\zM \subset \zV$, initial probability distribution $v$}
  \KwOut{An element $e \in \zM$}
   \textsc{Setup :} Sample a vertex $x$ according to $v$ and initialise the data structure.\;
   \Repeat{
     \textsc{Check :} \eIf{current vertex $x$ is marked}{
     \KwRet $x$\; 
   }{
     \RepeatU{$x$ is sampled according to a distribution close enough to the uniform distribution}{
       \textsc{Update :} \emph{Take one step of the random walk and update data structure accordingly.}\;
       }
     }
   }

  \caption{$RandomWalk$}
 \end{algorithm}

Let $\CS$ be the cost of \textsc{Setup}, $\CC$ be the cost of \textsc{Check} and $\CU$ be the cost of \textsc{Update}. 
It follows from the preceding considerations that $\tilde O(1/\delta)$ steps of the random walk are sufficient to sample $x$ according to the uniform distribution.
Furthermore, if we note $\varepsilon := \frac{|\zM|}{|\zV|}$ the proportion of marked elements, it is readily seen that the algorithm ends after $O(1/\varepsilon)$ iterations of the outer loop.
Thus the complexity of classical random walk is $\CS + \frac{1}{\varepsilon}\left(\CC + \frac{1}{\delta}\CU\right)$.

Several quantum versions of random walk algorithms have been proposed by many authors, notably Ambainis \cite{A07}, Szegedy \cite{S04}, and Magniez, Nayak, Roland and Santha \cite{MNRS07}. A survey of these results can be found in \cite{S08}. We use here the following result
\begin{theorem}[\cite{MNRS07}]
\label{th:quantumwalk}
Let $M$ be an aperiodic, irreducible and reversible Markov chain on a graph with spectral gap $\delta$, and $\varepsilon := \frac{|\zM|}{|\zV|}$ as above. Then there is a quantum walk algorithm that finds an element in $\zM$ with cost
\begin{equation}\label{eq:complexity_quantum_walk}
\boxed{\CS + \frac{1}{\sqrt\varepsilon}\left(\CC + \frac{1}{\sqrt\delta}\CU\right)}
\end{equation}
\end{theorem}

\subsubsection{Johnson graphs and product graphs.}
With the exception of Grover's search algorithm seen as a quantum walk algorithm, to date an overwhelming majority of quantum walk algorithms are based on Johnson graphs or a variant thereof. The decoding algorithms which shall be presented in this paper rely on cartesian products of Johnson graphs. All of these objects are defined in this section and some important properties are mentioned.
\begin{definition}[Johnson graphs]
\label{def_johnson_graph}
A Johnson graph $J(n,r)$ is an undirected graph whose vertices are the subsets containing $r$ elements of a set of size $n$, with an edge between two vertices $S$ and $S'$ iff $|S \cap S'| = r-1$. In other words, $S$ is adjacent to $S'$ if $S'$ can be obtained from $S$ by removing an element and adding a new element in its place.
\end{definition}
It is clear that $J(n,r)$ has $\binom{n}{r}$ vertices and is $r(n-r)$-regular. Its spectral gap is given by
\begin{equation}
\label{eq:spectral_gap_johnson}
\delta = \frac{n}{r(n-r)}.
\end{equation}
\begin{definition}[Cartesian product of graphs]
\label{def_product_graphs}
Let $G_1 = (\zV_1, \zE_1)$ and $G_2 = (\zV_2,\zE_2)$ be two graphs. Their cartesian product $G_1 \times G_2$ is the graph $G = (\zV,\zE)$ where:
\begin{enumerate}
  \item $\zV = \zV_1 \times \zV_2$, i.e. $\zV = \{v_1v_2~|~v_1 \in \zV_1, v_2 \in \zV_2\}$
  \item $\zE = \{(u_1u_2,v_1v_2)~|~ (u_1 = v_1 \wedge (u_2,v_2) \in \zE_2) \vee ((u_1,v_1) \in \zE_1 \wedge u_2 = v_2)\}$
\end{enumerate}
\end{definition}
The spectral gap of products of Johnson graphs is given by
\begin{restatable}[Cartesian product of Johnson graphs]{theorem}{thmProdJohnsonGraph}
\label{thm:product_johnsongraphs}
Let $J(n,r) = (\zV,\zE)$, $m \in \mathbb{N}$ and $J^m(n,r) := \times_{i=1}^m J(n,r) = (\zV_m,\zE_m)$. Then:
\begin{enumerate}
  \item $J^m(n,r)$ has $\binom{n}{r}^m$ vertices and is $md$-regular where $d = r(n-r)$.
  \item We will write $\delta(J)$ resp. $\delta(J^m)$ for the spectral gaps of $J(n,r)$ resp. $J^m(n,r)$. Then:\\
$\delta(J^m) \geq \frac{1}{m}\delta(J).$
 \item The random walk associated with $J^m(n,r)$ is  aperiodic, irreducible and reversible for all positive $m$, $n$ and $r<n$.
\end{enumerate}
\end{restatable}
For a proof of this statement, see the appendix.

\section{Generalities on classical and quantum decoding}
\label{sec:classical_quantum_decoding}
We first recall how the simplest ISD algorithm \cite{P62} and its quantised version \cite{B10} work and 
then give a skeleton of the structure of more sophisticated classical and quantum versions.
\subsection{Prange's algorithm and Bernstein's algorithm}
Recall that the goal is to find $e$ of weight $w$ given $s^T = He^T$, where $H$ is an $(n-k) \times n$ matrix. In other words, the problem we aim to solve is finding a solution to an underdetermined linear system of $n-k$ equations in $n$ variables and the  solution is unique owing to the weight condition. Prange's algorithm is based on the following observation: if it is known that $k$ given components of the error vector are zero, the error positions are among the $n-k$ remaining components. In other words, if we know for sure that the $k$ corresponding variables are not involved in the linear system, then the error vector can be found by solving the resulting linear system of $n-k$ equations in $n-k$ variables in polynomial time.

The hard part is finding a correct size-$k$ set (of indices of the components). Prange's algorithm 
samples such sets and solves the resulting linear equation until an error vector of weight $w$ is found.
The probability for finding such a set is of order 
$\Om{\frac{\binom{n-k}{w}}{\binom{n}{w}}}$ and therefore
Prange's algorithm has  
complexity 
$$\OO{ \frac{\binom{n}{w}}{\binom{n-k}{w}}}=\OOt{2^{\alpha_{\text{Prange}}(R,\omega)n}}$$
where
$$\alpha_{\text{Prange}}(R,\omega) = H_2(\omega) - (1-R)H_2\left(\frac{\omega}{1-R}\right)$$
by using the well known formula for binomials
$$
\binom{n}{w} = \Tht{2^{H_2\left( \frac{w}{n} \right)n}}.
$$
Bernstein's algorithm consists in using Grover's algorithm to find a correct size-$k$ set. Indeed, an oracle for checking that a size-$k$ set is correct can be obtained by following the same steps as in Prange's algorithm, i.e. deriving and solving a linear system of $n-k$ equations in $n-k$ variables and returning 1 iff the resulting error vector has weight $w$.
Thus the complexity of Bernstein's algorithm is the square root of that of Prange's algorithm, i.e. $\alpha_{\text{Bernstein}} = \frac{\alpha_{\text{Prange}}}{2}$.

\subsection{Generalised ISD algorithms}
More sophisticated classical ISD algorithms \cite{S88,D91,FS09,BLP11,MMT11,BJMM12,MO15} generalise Prange's algorithm in the following way: they introduce a new parameter $p$ and allow $p$ error positions 
inside of the size-$k$ set (henceforth denoted by $\zS$). Furthermore, from Dumer's algorithm onwards, a new parameter $\ell$ is introduced and the set $\zS$ is taken to be of size $k+\ell$.
This event happens with probability $P_{\ell,p} \eqdef \frac{\binom{k+\ell}{p}\binom{n-k-\ell}{w-p}}{\binom{n}{w}}.$
The point is that
\begin{restatable}{proposition}{propPuncturing}
\label{prop:puncturing}
Assume that the restriction of $H$ to the columns belonging to the complement of $\zS$ is a matrix of full rank, then
\begin{itemize}
\item[(i)]
the restriction $e'$ of the error to $\zS$  is a solution to the syndrome decoding problem
 \begin{equation}\label{eq:subproblem}
H' {e'}^T = {s'}^T.
\end{equation} 
with $H'$ being an $\ell \times (k+\ell)$ binary matrix, $|e'|=p$ and $H'$, $s'$ that can be computed in polynomial time from $\zS$, $H$ and $s$;\\
\item[(ii)] once we have such an $e'$,  there is a unique $e$ whose restriction to $\zS$ is equal to $e'$ and which satisfies
$He^T = s^T$. Such an $e$ can be computed from $e'$ in polynomial time.
\end{itemize}
\end{restatable}

\noindent
{\em Remark:} The condition in this proposition is met with very large probability when $H$ is chosen uniformly at random:
it fails to hold with probability which is only $O(2^{-\ell})$.

\begin{proof}
Without loss 
of generality assume that $\zS$ is given by the $k+\ell$ first positions.  By performing Gaussian elimination, we look for a square matrix $U$ such that 
$$
U H = \begin{pmatrix} H' & 0_\ell \\
H" & I_{n-k-\ell}\end{pmatrix}
$$
That such a matrix exists is a consequence of the fact that $H$ restricted to the last $n-k-\ell$ positions is of full rank.
Write now $e=(e',e")$ where $e'$ is the word formed by the $k+\ell$ first entries of $e$. Then
$$Us^T=UHe^T = \begin{pmatrix} H'{e'}^T \\ H"{e'}^T + {e"}^T \end{pmatrix}.$$
If we write $Us^T$ as $(s',s")^T$, where $s'^T$ is the vector formed by the $\ell$ first entries 
of $Us^T$, then we recover $e$ from $e'$
by using the fact that $H"{e'}^T + {e"}^T={s"}^T$. $~\qed$
\end{proof}
 From now on, we denote by $\Sigma$ and $h$ the functions that can be computed in polynomial time that are promised by this proposition, i.e.
\begin{eqnarray*}
s ' & = & \Sigma(s,H,\zS)\\
e & = & h(e')
\end{eqnarray*}

In other words, all these algorithms  solve in a first step a new instance of the syndrome decoding problem with different parameters. The difference
with the original problem is that if $\ell$ is small, which is the case in general, there is not a single solution anymore. 
However searching for all (or a large set of them) can be done more efficiently than just brute-forcing over all errors of weight $p$ on the set $\zS$.
Once a possible solution $e'$ to \eqref{eq:subproblem} is found, $e$ is recovered as explained before. 
The main idea which avoids brute forcing over all possible errors of weight $p$ on $\zS$ is to obtain candidates $e'$ by solving an instance of a
generalised $k$-sum problem that we define as follows.
\begin{problem}[generalised $k$-sum problem]
Consider an Abelian group $\zG$, an arbitrary set $\zE$, a map $f$ from $\zE$ to $\zG$, $k$ subsets $\zV_0$, $\zV_1$, \dots, $\zV_{k-1}$ of $\zE$, another map $g$ from $\zE^k$ to $\{0,1\}$,  and an element $S \in \zG$.
 Find
a  solution $(v_0,\dots,v_{k-1}) \in \zV_0\times \dots \zV_{k-1}$ such that we have at the same time
\begin{itemize}
\item[(i)] $f(v_0) + f(v_1) \dots + f(v_{k-1}) = S$ (subset-sum condition);
\item[(ii)] $g(v_0,\dots,v_{k-1})  =  0$   $((v_0,\dots,v_{k-1})$ is a root of $g)$.
\end{itemize}
\end{problem}

Dumer's ISD algorithm, for instance, solves the $2$-sum problem in the case where
\begin{eqnarray*}
\zG & = &\Ft^\ell, \;\;\zE = \Ft^{k+\ell},\;\;f(v)  =  H'{v}^T\\
\zV_0 &= &\{(e_0,0_{(k+\ell)/2})\in \Ft^{k+\ell} : e_0 \in \Ft^{(k+\ell)/2},\; |e_0|=p/2\}  \\
\zV_1 &= &\{(0_{(k+\ell)/2},e_1)\in \Ft^{k+\ell} : e_1 \in \Ft^{(k+\ell)/2},\; |e_1|=p/2\} 
\end{eqnarray*}
and $g(v_0,v_1)=0$ if and only if $e=h(e')$ is of weight $w$ where $e'=v_0 + v_1$.
 A solution to the $2$-sum problem is then clearly a solution to the decoding problem by construction.
The point is that the $2$-sum problem can be solved in time which is much less than $|\zV_0|\cdot|\zV_1|$. For instance,
this can clearly be achieved in expected time $|\zV_0|+|\zV_1|+\frac{|\zV_0|\cdot|\zV_1|}{|\zG|}$
and space $|\zG|$ by storing the elements $v_0$ of $\zV_0$ in a hashtable at the address $f(v_0)$ and then going over all elements $v_1$ of the other set 
to check whether or not the address $S-f(v_1)$ contains an element. The term $\frac{|\zV_0|\cdot|\zV_1|}{|\zG|}$ accounts for the expected number of solutions 
of the $2$-sum problem when the elements of $\zV_0$ and $\zV_1$ are chosen uniformly at random in $\zE$ (which is the assumption what we are going to make from on).
This is precisely what Dumer's algorithm does. Generally, the size of $\zG$ is chosen such that $|\zG|=\Th{|\zV_i|}$ and the space and time complexity are
also of this order.

Generalised ISD algorithms are thus composed of a loop in which first a set $\zS$ is sampled and then an error vector having a certain form, namely with $p$ error positions in $\zS$ and $w-p$ error positions outside of $\zS$, is sought. Thus, 
for each ISD algorithm $A$, we will denote by $Search_A$ the algorithm
 whose exact implementation depends on $A$ but whose specification is always\\
$Search_A : \zS, H, s, w, p \rightarrow \{ e ~|~ \text{$e$ has weight $p$ on $\zS$ and weight $w-p$ on $\overline{\zS}$ and } s^T = He^T\} \cup \{NULL\}$, 
 where $\zS$ is a set of indices, $H$ is the parity-check matrix of the code and $s$ is the syndrome of the error we are looking for.
The following pseudo-code gives the structure of a generalised ISD algorithm.\\
\begin{algorithm}[H]
 \DontPrintSemicolon
 \KwIn{$H$, $s$, $w$, $p$}
 \KwOut{$e$ of weight $w$ such that $s^T = He^T$}
 \RepeatU{$|e| = w$}{
  Sample a set of indices $\zS \subset \{1, ...,n\}$\;
  $e \leftarrow Search_A(\zS, H, s, w, p)$\;
 }
 \KwRet $e$\;
 \caption{ISD\_Skeleton}
\end{algorithm}
$~$
Thus, if we note $P_A$ the probability, dependent on the algorithm $A$, that the sampled set $\zS$ is correct and that $A$ finds $e$
\footnote{In the case of Dumer's algorithm, for instance, even if the restriction of $e$ to $\zS$ is of weight $p$, Dumer's algorithm
may fail to find it since it does not split evenly on both sides of the bipartition of $\zS$.}, and $T_A$ the execution time of the algorithm $Search_A$, the complexity of generalised ISD algorithms is 
$\OO{\frac{T_A}{P_A}}.$
To construct generalised quantum ISD algorithms, we use Bernstein's idea of using Grover search to look for a correct set $\zS$. 
However, now each query made by Grover search will take time which is essentially the time complexity of $Search_A$.
Consequently, the complexity of generalised quantum ISD algorithms is given by the following formula:
\begin{equation}
\label{eq:T_quant_ISD}
\OO{\frac{T_A}{\sqrt{P_A}}} = \OO{\sqrt{\frac{T_A^2}{P_A}}}.
\end{equation}
An immediate consequence of this formula is that, in order to halve the complexity exponent of a given classical algorithm, we need a quantum algorithm whose search subroutine is ``twice'' as efficient.

{\tiny }\section{Solving the generalised $4$-sum problem with quantum walks and Grover search}

\subsection{The Shamir-Schroeppel idea}
As explained in Section \ref{sec:classical_quantum_decoding}, the more sophisticated ISD algorithms solve during the inner step 
an instance of the generalised $k$-sum problem.
The issue is to get a good quantum version of the classical algorithms used to solve this problem. That this task is non trivial can already be
guessed from Dumer's algorithm. Recall that it solves the generalised $2$-sum problem in time and space complexity $\OO{V}$ when 
$V=|\zV_0|=|\zV_1|=\Theta(|\zG|)$.
The problem is that if we wanted a quadratic speedup when compared to the classical Dumer algorithm,
then this would require a quantum algorithm solving the same problem in time $\OO{V^{1/2}}$, but this seems problematic since naive ways of quantising this algorithm 
stumble on the problem that the space complexity is a lower bound on the time complexity of the quantum algorithm.
This strongly motivates the choice of ways of solving the $2$-sum problem by using less memory. This can be done through the idea of Shamir and Schroeppel \cite{SS81}.
Note that the very same idea is also used for the same reason to 
speed up quantum algorithms for the subset sum problem in \cite[Sec. 4]{BJLM13}.
To explain the idea, suppose that $\zG$ factorises as $\zG = \zG_0 \times \zG_1$ where $|\zG_0| = \Theta(|\zG_1|)=\Theta(|\zG|^{1/2})$.
Denote for $i\in \{0,1\}$ by $\pi_i$ the projection from $\zG$ onto $\zG_i$ which to $g=(g_0,g_1)$ associates  $g_i$.

The idea is to construct $f(\zV_0)$ and $f(\zV_1)$ themselves as 
$f(\zV_0)=f(\zV_{00})+f(\zV_{01})$ and 
$f(\zV_1)=f(\zV_{10}) + f(\zV_{11})$ in such a way that  the $\zV_{ij}$'s are of size $O(V^{1/2})$
and to solve a $4$-sum problem by solving various $2$-sum problems. In our coding theoretic setting, it will be more convenient to explain everything directly in terms
of the $4$-sum problem which is given in this case by
\begin{problem}\label{pb:SS}
Assume that $k+\ell$  and $p$ are multiples of $4$. 
Let 
\begin{eqnarray*}
\zG & = &\Ft^\ell, \;\;\zE  =  \Ft^{k+\ell}, \;\;f(v)  =  H'{v}^T\\
\zV_{00} &\eqdef &\{(e_{00},0_{3(k+\ell)/4})\in \Ft^{k+\ell} : e_{00 }\in \Ft^{(k+\ell)/4},\; |e_{00}|=p/4\} \\
\zV_{01} &\eqdef &\{(0_{(k+\ell)/4},e_{01},0_{(k+\ell)/2})\in \Ft^{k+\ell} : e_{01 }\in \Ft^{(k+\ell)/4},\; |e_{01}|=p/4\} \\
\zV_{10} &\eqdef &\{(0_{(k+\ell)/2},e_{10},0_{(k+\ell)/4})\in \Ft^{k+\ell} : e_{10 }\in \Ft^{(k+\ell)/4},\; |e_{10}|=p/4\} \\
\zV_{11} &\eqdef &\{(0_{3(k+\ell)/4},e_{11})\in \Ft^{k+\ell} : e_{11}\in \Ft^{(k+\ell)/4},\; |e_{11}|=p/4\} 
\end{eqnarray*}
and $S$ be some element in $\zG$.
Find $(v_{00},v_{01},v_{10},v_{11})$ in $\zV_{00}\times \zV_{01} \times \zV_{10} \times \zV_{11}$ such 
that $f(v_{00})+f(v_{01})+f(v_{10})+f(v_{11})=S$ and $h(v_{00}+v_{01}+v_{10}+v_{11})$ is of weight $w$.
\end{problem}
Let us explain now how the Shamir-Schroeppel idea allows us to solve the $4$-sum problem in time $\OO{V}$ and space $\OO{V^{1/2}}$ when 
the $\zV_{ij}$'s are of order $\OO{V^{1/2}}$, 
$|\zG|$ is of order $V$ and when $\zG$ decomposes as the product of two groups $\zG_0$ and $\zG_1$ both of size $\Th{V^{1/2}}$.
The basic idea is  to solve for all possible $r \in \zG_1$ the following $2$-sum problems 
\begin{eqnarray}
\pi_1(f(v_{00})) +\pi_1(f(v_{01})) & = & r\label{eq:problem1}\\
\pi_1(f(v_{10}))  +\pi_1(f(v_{11})) & = & \pi_1(S)-r \label{eq:problem2}
\end{eqnarray}
Once these problems are solved we are left with  $\OO{V^{1/2} V^{1/2}/V^{1/2}}=\OO{V^{1/2}} $ solutions to the first problem and $\OO{V^{1/2}}$ solutions to the second.
Taking any pair $(v_{00},v_{01})$ solution to \eqref{eq:problem1}  and $(v_{10},v_{11})$ solution to \eqref{eq:problem2} yields a $4$-tuple which is a partial solution to the
$4$-sum problem
$$
\pi_1( f(v_{00}))+\pi_1(f(v_{01})) +\pi_1(f(v_{10}))+\pi_1(f(v_{11}))  = r + \pi_1(S)-r = \pi_1(S).
$$
Let $\zV'_0$ be the set of all pairs $(v_{00},v_{01})$ we have found for the 
first $2$-sum problem \eqref{eq:problem1}, whereas $\zV'_1$ is the set of all solutions to \eqref{eq:problem2}. 
To ensure that $f(v_{00})+f(v_{01}) +f(v_{10})+f(v_{11}) =  S$ 
we just have to 
solve the following $2$-sum problem
$$
\underbrace{\pi_0(f(v_{00})) + \pi_0(f(v_{01}) )}_{f'(v_{00},v_{01})}+ \underbrace{\pi_0(f(v_{10})) + \pi_0(f(v_{11}) )}_{f'(v_{10},v_{11})} = \pi_0(S)
$$
and $$
g(v_{00},v_{01},v_{10},v_{11})=0
$$
where $(v_{00},v_{01})$ is in $\zV'_0$, $(v_{10},v_{11})$ is in $\zV'_1$ and 
$g$ is the function whose root we want to find for the original $4$-sum problem.

This is again of complexity $\OO{V^{1/2} V^{1/2}/V^{1/2}}=\OO{V^{1/2}} $.
Checking a particular value of $r$ takes therefore $\OO{V^{1/2}}$
 operations. Since we have $\Th{V^{1/2}}$ values to check, the 
total complexity is $\OO{V^{1/2}V^{1/2}}=\OO{V}$, that is the same as before, but we need only $\OO{V^{1/2}}$ memory to store all intermediate sets.
\begin{figure}[h]
    \centering
    \includegraphics[width=0.5\textwidth]{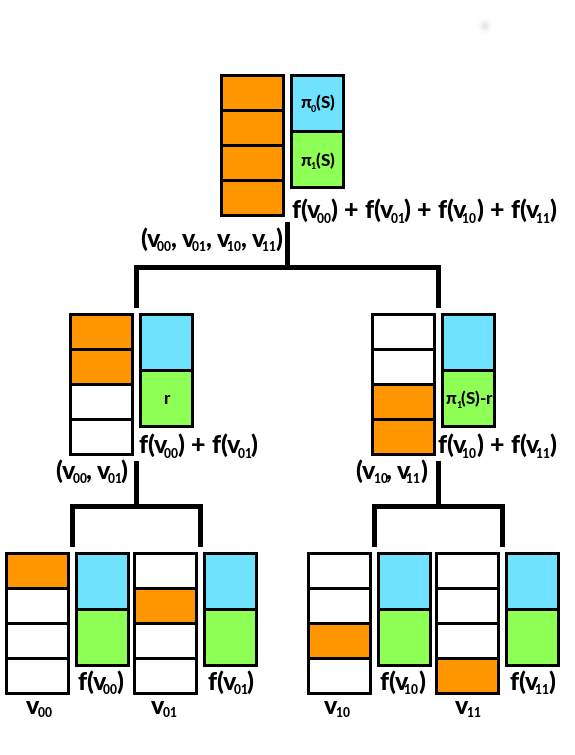}
    \caption{
The Shamir-Schroeppel idea in the decoding context (see Problem \ref{pb:SS}): the support of the elements of $\zV_{ij}$ is represented in orange, while
the blue and green colours represent $\zG_0$ resp. $\zG_1$.}
    \label{fig:ss}
\end{figure}
\subsection{A quantum version of the Shamir-Schroeppel algorithm}

By following the approach of \cite{BJLM13}, we will define a quantum
algorithm for  
solving the $4$-sum problem
by combining Grover search with a quantum walk
with a complexity given by

\begin{proposition}\label{prop:easy}
Consider the generalised $4$-sum problem  with sets $\zV_u$ of size  $V$. Assume that 
$\zG$ can be decomposed as $\zG =\zG_0 \times \zG_1$. 
There is a quantum algorithm for solving the $4$-sum problem running in time $\OOt{|\zG_1|^{1/2}V^{4/5}}$
as soon as $|\zG_1| = \Om{V^{4/5}}$ and $|\zG| = \Om{V^{8/5}}$.
\end{proposition}

This is nothing but the idea of the algorithm \cite[Sec. 4]{BJLM13} laid out in a more general context.
The idea is as in the classical algorithm to look for the right value $r \in \zG_1$. This can be done with Grover search in 
time $\OO{|\zG_1|^{1/2}}$ instead of $\OO{|\zG_1|}$ in the classical case.
The quantum walk is then used to solve the following problem:

\begin{problem}\label{pb:shamir_shroeppel}
Find $(v_{00},v_{01},v_{10},v_{11})$ in $\zV_{00}\times \zV_{01} \times \zV_{10} \times \zV_{11}$ such that 
\begin{eqnarray*}
\pi_1(f(v_{00})) +\pi_1(f(v_{01})) & = & r \\
\pi_1(f(v_{10}))  +\pi_1(f(v_{11})) & = & \pi_1(S)-r \\
\pi_0(f(v_{00})) + \pi_0(f(v_{01})) + \pi_0(f(v_{10})) + \pi_0(f(v_{11})) & = &\pi_0(S)\\
g(v_{00},v_{01},v_{10},v_{11}) & = & 0.
\end{eqnarray*}
\end{problem}

For this, we choose subsets $\zU_i$'s of the $\zV_i$'s of a same size $U=\Th{V^{4/5}}$ and 
run a quantum walk on the graph whose vertices are all possible $4$-tuples of
sets of this kind and two $4$-tuples $(\zU_{00},\zU_{01},\zU_{10},\zU_{11})$ and
$(\zU_{00}',\zU_{01}',\zU_{10}',\zU_{11}')$ are adjacent if and only if we have for all $i$'s but one $\zU'_i=\zU_i$ and
for the remaining  $\zU'_i$ and $\zU_i$ we have $|\zU'_i \cap \zU_i|=U-1$.
Notice that this graph is nothing but $J^4(V,U)$.
By following \cite[Sec. 4]{BJLM13} it can be proved that 
\begin{restatable}{proposition}{propDataStructure}
\label{prop:data_structure}
Under the assumptions that $|\zG_1| = \Om{V^{4/5}}$ and $|\zG| = \Om{V^{8/5}}$,
it is possible to set up a data structure of size $\OO{U}$
to implement this quantum walk such that \\
(i) setting up the data structure takes time $\OO{U}$;\\
(ii) checking whether a new $4$-tuple leads to a solution to the problem above (and outputting the solution in this case) 
takes time $\OO{1}$,\\
(iii) updating the data structure takes time $\OO{\log U}$.
\end{restatable}
The proof which we give is adapted from \cite[Sec. 4]{BJLM13}.
\begin{proof}
$~$\\
\begin{enumerate}
\item \textbf{Setting up the data structure takes time $\OO{U}$.\\}
The data structure is set up more or less in the same way as in classical Shamir-Schroeppel's algorithm, i.e. by solving two 2-sum problems first and then using the result to solve a third and last 2-sum problem. There are however the following differences:\\ \\
(i) We no longer keep the results in a hashtable but in some other type of ordered data structure which allows for the insertion, deletion and search operations to be done in $\OO{\log U}$ time. For instance, \cite{BJLM13} chose radix trees. More detail will be given when we look at the \textsc{Update} operation.\\
(ii) Because we no longer use hashtables, we will need two data structures at each step, one to keep track of $f(v_{00})$ along with the associated $v_{00}$, $f(v_{00}) + f(v_{01})$ along with the associated $(v_{00},v_{01})$, etc. and another to keep track of $v_{00}$, $(v_{00}, v_{01})$, etc. separately. If we denote the first family of data structures by $\zD_f$ and the second family by $\zD_{\zV}$, this gives a total of 13 data structures (7 of type $\zD_{\zV}$ and 6 of type $\zD_f$, because no data structure is needed to  
store the sum of all four vectors which is simply $S$).\\ \\
Solving the first two 2-sum problems takes time $|\zU_{i0}| + |\zU_{i1}| + \frac{|\zU_{i0}|.|\zU_{i1}|}{|\zG_1|}$,
$i=0,1$, which is $\OO{U}$ because $|\zG_1| = \Om{V^{4/5}} = \Om{U}$. 
Denote by $\zU_{0}$ resp. $\zU_{1}$ the set of solutions to these two problems. These solutions are used to solve the second 2-sum, problem, which takes time 
$|\zU_{0}| + |\zU_{1}| + \frac{|\zU_{0}|.|\zU_{1}|}{|\zG_0|} = \OO{U}$ 
due to 
$\zG_0 = \zG / \zG_1$ 
and $|\zG| = \Om{V^{8/5}}$.\\ \\
Thus, setting up the data structure takes time $\OO{U}$.
\item \textbf{Updating the data structure takes time $\OO{\log U}$.\\}
Recall that the data structures are chosen such that the insertion, deletion and search operations take $\OO{\log U}$ time, and also that there are two data structures pertaining to each vector or pair of vectors, for a total of 13 data structures.\\ \\
Recall also that the update step consists in moving from one vertex of the Johnson graph $J^4(V,U)$ to one that is adjacent to it. Suppose, without loss of generality, that we move from the vertex $(\zU_{00},\zU_{01},\zU_{10},\zU_{11})$ to $(\zU_{00}',\zU_{01},\zU_{10},\zU_{11})$. Thus, a $v_{00} \in \zU_{00}$ has been replaced by a $u_{00} \in \zU_{00}$.\\ \\
Then, the low cost of the update step relies upon the following fundamental insight: there are in all $U$ possible ways of writing the sum $\pi_1(f(u_{00})) +\pi_1(f(v_{01}))$ (one for each $v_{01} \in \zU_{01}$). But we have one further constraint which is that this sum needs to be equal to a given $r \in \zG_1$. Thus, there are on average $\frac{|\zU_{ij}|}{|\zG_1|} = O(1)$ values of $v_{01} \in \zU_{01}$ which give a solution.\\ \\
Note that the same argument applies for the number of $(v_{10}, v_{11}) \in \zU_{1}$ that fulfil the condition
$$\pi_0(f(u_{00})) + \pi_0(f(v_{01})) + \pi_0(f(v_{10})) + \pi_0(f(v_{11})) ~ = ~\pi_0(S)$$
for a given 
$(u_{00}, v_{01}) \in \zU_{0}$
(where $\pi_0(S) \in \zG_0$), 
for in this case there are on average 
$\frac{|\zU_{0}|}{|\zG_0|} = O(1)$ 
such elements.\\ \\
This allows us to proceed as follows: we impose a constant limit on the number of $v_{01} \in \zU_{01}$ that correspond to a given $u_{00} \in \zU_{00}$ at each update operation. A similar limit is imposed on the number of $(v_{10}, v_{11}) \in \zU_{1}$. The probability of reaching this limit is negligeable, and if it is reached, we re-initialise the data structure, so this does not modify the overall complexity of the algorithm. Note also that there is no problem when the opposite situation happens, i.e. when there are no $v_{01} \in \zU_{01}$ corresponding to a given $u_{00} \in \zU_{00}$. Indeed, while this may result in the data structure being depleted, this is a temporary situation and the data structure will be refilled over time as more suitable elements occur.\\ \\
We now enumerate the steps needed to update the data structure. What we need to do is to remove the old element $v_{00}$ and everything that has been constructed using it, and add $u_{00}$ and everything that it allows to construct (within the limits discussed above).
First, to remove $v_{00}$ and the other elements it affects, we need to do the following:
\begin{enumerate}
\item Find and delete $v_{00}$ from the data structure $\zD_{\zU_{00}}$.
\item Calculate $f(v_{00})$, then find and delete it from the data structure $\zD_{f_{00}}$.
\item Find at most a constant number of $(v_{00}, v_{01})$ in $\zD_{\zU_{0}}$ and remove them.
\item For each of these $(v_{00}, v_{01})$, calculate $f(v_{00}) + f(v_{01})$ and remove it from $\zD_{f_{0}}$.
\item Find at most a constant number of $(v_{00}, v_{01}, v_{01}, v_{11})$ in $\zD_{\zU}$ and remove them.
\end{enumerate}
This step uses operations of negligeable cost (calculating $f(v_{00})$, etc.) and the number of operations of cost 
$\log(U)$ 
which it uses is bounded by a constant. Thus it takes time 
$\OO{\log(U)}$.\\ \\
To add $u_{00}$ and other new elements depending on it, we proceed as follows:
\begin{enumerate}
\item Insert $u_{00}$ in $\zD_{\zU_{00}}$.
\item Calculate $f(u_{00})$, then insert it in $\zD_{f_{00}}$.
\item Calculate $x=r - \pi_1(f(u_{00}))$ and find if there are elements $y$ in $\zD_{f_{01}}$ such that $\pi_1(y)=x$. For a constant number of associated $v_{01}$, insert $(u_{00}, v_{01})$ in $\zD_{\zU_{0}}$ and in $\zD_{f_{0}}$ associated with $r$.
\item Similarly there are a constant number of $(v_{01}, v_{11})$ that need to be updated, for those calculate $g(u_{00}, v_{01}, v_{10}, v_{11})$. If it is equal to zero, insert $(v_{00}, v_{01}, v_{10}, v_{11})$ in $\zD_{\zU}$.
\end{enumerate}
It is easy to see that this step also takes time 
$\OO{\log(U)}$.
\item \textbf{Checking whether a new $4$-tuple leads to a solution of the problem takes time $\OO{1}$.\\}
Checking that the right $4$-tuple is in $\zD_{\zU}$ requires looking for it in $\zD_{\zU}$ at the first step of the algorithm. This costs $O(\sqrt{U})$ using Grover search. At the following steps of the algorithm, it is enough to check the new elements (whose number is bounded by a constant) that have been added to $\zD_{\zU}$. So the checking cost is 
$\OO{1}$
 overall.$~\qed$
\end{enumerate}
\end{proof}
Proposition \ref{prop:easy} is essentially a corollary of this proposition. 

\begin{proof} [Proof of Proposition \ref{prop:easy}]
Recall that 
the cost of the quantum walk is given by
$
\CS + \frac{1}{\sqrt\varepsilon}\left(\CC + \frac{1}{\sqrt\delta}\CU\right)$
 where $\CS,\CC,\CU,\varepsilon$ and $\delta$
are the setup cost, the check cost, the update cost, the proportion of marked elements
and the spectral gap of the quantum walk. 
From Proposition \ref{prop:data_structure}, we know that
$\CS =  \OO{U} = \OO{V^{4/5}}$, $\CC  =  \OO{1}$,
and $\CU  =  \OO{\log U}$.
Recall that the spectral gap of $J(V,U)$ is equal to $\frac{V}{U(V-U)}$ by \eqref{eq:spectral_gap_johnson}.
This quantity is larger than $\frac{1}{U}$ and by using  
 Theorem \ref{thm:product_johnsongraphs} on the cartesian product of Johnson graphs, we obtain
$\delta = \Th{\frac{1}{U}}$.

Now for the proportion of marked elements we argue as follows. If Problem \ref{pb:shamir_shroeppel} has a solution $(v_{00},v_{01},v_{10},v_{11})$, then the probability that 
each of the sets $\zU_i$ contains $v_i$ is precisely $U/V=\Th{V^{-1/5}}$.
The probability $\varepsilon$ that all the $\zU_i$'s contain $v_i$ is then 
$\Th{V^{-4/5}}$.  
This gives a total cost of 
$$
\OO{V^{4/5}} + \OO{V^{2/5}}\left( \OO{1} + \OO{V^{2/5}}\OO{\log U}\right) = \OOt{V^{4/5}}.
$$
When we multiply this by the cost of Grover's algorithm for finding the right $r$ we have the aforementioned complexity.$~\qed$
\end{proof}

\subsection{Application to the decoding problem}
When applying this approach to the decoding problem we obtain
\begin{restatable}{theorem}{thmExpSSQW}
\label{th:expSSQW}
We can decode $w=\omega n$ errors in a random linear code of length $n$ and rate $R=\frac{k}{n}$  with a quantum complexity of
order $\OOt{2^{\alpha_{\text{SSQW}}(R,\omega)n}}$ where 
$$
\alpha_{\text{SSQW}}(R, \omega) \eqdef \min_{(\pi,\lambda) \in \zR} \left(\frac{H_2(\omega) - (1-R-\lambda)H_2\left(\frac{\omega - \pi}{1 - R - \lambda}\right) - \frac{2}{5}(R+\lambda)H_2\left(\frac{\pi}{R+\lambda}\right)}{2} \right)$$
$$
\zR  \eqdef  \left\{\!(\pi,\lambda)\!\!\in\!\![0,\omega]\!\times\![0,1)\!:\lambda = \frac{2}{5}(R+\lambda)H_2\!\left(\!\frac{\pi}{R+\lambda}\!\right)\!\!,
\pi \leq R + \lambda, \lambda \leq 1-R-\omega+\pi \! \right\}
$$
\end{restatable}
\begin{proof}
Recall (see \eqref{eq:T_quant_ISD}) that the quantum complexity is given by
\begin{equation}
\label{eq:complexitySSQW}
\tilde O\left(\frac{T_{\text{SSQW}}}{\sqrt{P_{\text{SSQW}}}}\right)
\end{equation}
where $T_{\text{SSQW}}$ is the complexity of the combination of Grover's algorithm and quantum walk
solving the generalised $4$-sum problem specified in Problem \ref{pb:shamir_shroeppel} and $P_{\text{SSQW}}$ is the probability
that the random set of $k + \ell$ positions $\zS$ and its random partition in $4$ sets of the same size that are chosen is such that 
all four of them contain exactly $p/4$ errors.
Note that $p$ and $\ell$ are chosen such that $k+\ell$ and $p$ are divisible by $4$.
$P_{\text{SSQW}}$ is given by
$$
P_{\text{SSQW}} = \frac{  {\binom{\frac{k+\ell}{4}}{\frac{p}{4}}}^4 \binom{n-k-\ell}{w - p}}{\binom{n}{w}}
$$
Therefore
\begin{equation}
\label{eq:PSSQW}
(P_{\text{SSQW}})^{-1/2} = \OOt{2^{\frac{H_2(\omega) - (1-R-\lambda)H_2\left(\frac{\omega - \pi}{1 - R - \lambda}\right) - (R+\lambda)H_2\left(\frac{\pi}{R+\lambda}\right) }{2} n}}
\end{equation}
where $\lambda \eqdef \frac{\ell}{n}$ and $\pi \eqdef \frac{p}{n}$. $T_{\text{SSQW}}$ is given by Proposition \ref{prop:easy}:
$$
T_{\text{SSQW}}  =  \OOt{|\zG_1|^{1/2}V^{4/5}}
$$
where the sets involved in the generalised $4$-sum problem are specified in Problem \ref{pb:shamir_shroeppel}.
This gives
$$
V  =  \binom{\frac{k + \ell}{4}}{\frac{p}{4}}
$$
We choose $\zG_1$ as 
\begin{equation}\label{eq:condition1}
\zG_1 = \Ft^{\lceil \frac{\ell}{2} \rceil}
\end{equation}
and the assumptions of Proposition \ref{prop:easy} are verified as soon as
$$
2^\ell = \Om{V^{8/5}}.
$$
which amounts to
$$
2^\ell = \Om{  {\binom{\frac{k + \ell}{4}}{\frac{p}{4}}}^{8/5}}
$$
This explains the condition 
\begin{equation}\label{eq:condition2}
\lambda = \frac{2}{5} (R + \lambda) H_2\left( \frac{\pi}{R+\lambda} \right)
\end{equation}
found in the definition of the region $\zR$.
With the choices \eqref{eq:condition1} and \eqref{eq:condition2}, 
we obtain
\begin{eqnarray}
T_{\text{SSQW}} &= &\OOt{V^{6/5}} \nonumber \\
& = & \OOt{ 2^{\frac{3}{10}(R+\lambda)H_2\left( \frac{\pi}{R+\lambda} \right)n}}\label{eq:TSSQW}
\end{eqnarray}
Substituting for $P_{\text{SSQW}}$ and $T_{\text{SSQW}}$ the expressions given by \eqref{eq:PSSQW} and \eqref{eq:TSSQW} finishes the proof of the
theorem. $~\qed$
\end{proof}

\section{Improvements obtained by the representation technique and ``$1+1=0$''}
There are two techniques that can be used to
speed up the quantum algorithm of the previous section.

\par{\em The representation technique.} It was introduced in \cite{HJ10} to speed up algorithms for the subset-sum algorithm and
used later on in \cite{MMT11} to improve decoding algorithms. The basic idea of the representation technique in the context of the 
subset-sum or decoding algorithms consists in (i) changing slightly the underlying (generalised) $k$-sum problem which is solved by 
introducing sets $\zV_i$ for which there are (exponentially) many solutions to the problem $\sum_i f(v_i) = S$ by using redundant representations,
 (ii) noticing that this allows us to put additional subset-sum conditions on the solution.

In the decoding context, instead of considering sets of errors with non-overlapping support, the idea that allows us to obtain many 
different representations of a same solution is just to consider sets $\zV_i$ corresponding to errors with overlapping supports. In our case,
we could have taken instead of the four sets defined in the previous section the following sets
\begin{eqnarray*}
\zV_{00}=\zV_{10} &\eqdef &\{(e_{00},0_{(k+\ell)/2})\in \Ft^{k+\ell} : e_{00 }\in \Ft^{(k+\ell)/2},\; |e_{00}|=p/4\} \\
\zV_{01}=\zV_{11} &\eqdef &\{(0_{(k+\ell)/2},e_{01})\in \Ft^{k+\ell} : e_{01 }\in \Ft^{(k+\ell)/2},\; |e_{01}|=p/4\}
\end{eqnarray*}
Clearly a vector $e$ of weight $p$ can be written in many different ways as a sum
$v_{00}+v_{01}+v_{10}+v_{11} $ where $v_{ij}$ belongs to $\zV_{ij}$. This is (essentially) due to the fact that a vector of weight $p$ can be written in $\binom{p}{p/2} = \OOt{2^p}$ 
ways as a sum of two vectors  of weight $p/2$.

The point is that if we apply now the same algorithm as in the previous section and look for solutions to Problem \ref{pb:SS}, there is not a single value of $r$
that leads to the right solution. Here, about $2^p$ values of $r$ will do the same job. The speedup obtained by the representation technique 
is a consequence of this phenomenon. We can even improve on this representation technique by using the $1+1=0$ phenomenon as in \cite{BJMM12}.

\par{\em The ``$1+1=0$'' phenomenon.} Instead of choosing the $\zV_i$'s as explained above we will actually choose the $\zV_i$'s as
\begin{eqnarray}
\zV_{00}=\zV_{10} &\eqdef &\{(e_{00},0_{(k+\ell)/2})\in \Ft^{k+\ell} : e_{00 }\in \Ft^{(k+\ell)/2},\; |e_{00}|=\frac{p}{4} +\frac{\Delta p}{2}\} \label{eq:v00v10}\\
\zV_{01}=\zV_{11} &\eqdef &\{(0_{(k+\ell)/2},e_{01})\in \Ft^{k+\ell} : e_{01 }\in \Ft^{(k+\ell)/2},\; |e_{01}|=\frac{p}{4} +\frac{\Delta p}{2}\} \label{eq:v01v11}
\end{eqnarray}
A vector $e$ of weight $p$ in $\Ft^{k+\ell}$ can indeed by represented in many ways as a sum of $2$ vectors of weight $\frac{p}{2} +\Delta p$. More precisely, such a vector can be 
represented in 
$
\binom{p}{p/2} \binom{k+\ell-p}{\Delta p} 
$
ways. Notice that this number of representations is greater than the number $2^p$ that we had before. This explains why choosing an appropriate
positive value $\Delta p$ allows us to improve on the previous choice.

The quantum algorithm for decoding follows the same pattern as in the previous section:
(i) we look with Grover's search algorithm for a right set $\zS$ of $k + \ell$ positions 
such that the restriction $e'$ of the error $e$ we look for is of weight $p$ on this subset
and then (ii) we search for $e'$ by solving a generalised $4$-sum problem  with a combination
of Grover's algorithm and a quantum walk. We will use for the second point the following proposition
which quantifies how much we gain when there are multiple representations/solutions:

\begin{proposition}\label{prop:improvement}
Consider the generalised $4$-sum problem  with sets $\zV_u$ of size  $\OO{V}$. Assume that 
$\zG$ can be decomposed as $\zG =\zG_0 \times \zG_1 \times \zG_2$. Furthermore assume that 
there are $\Om{|\zG_2|}$ solutions to the $4$-sum problem and that we can fix arbitrarily 
the value $\pi_2\left(f(v_{00})+f(v_{01})\right)$ of a solution to the $4$-sum problem, 
where $\pi_2$ is the mapping from $\zG =\zG_0 \times \zG_1 \times \zG_2$ to $\zG_2$ which maps 
$(g_0,g_1,g_2)$ to $g_2$.
There is a quantum algorithm for solving the $4$-sum problem running in time $\OOt{|\zG_1|^{1/2}V^{4/5}}$
as soon as $|\zG_1|\cdot|\zG_2| = \Om{V^{4/5}}$ and $|\zG| = \Om{V^{8/5}}$.
\end{proposition}

\begin{proof}
Let us first introduce a few notations.
We denote by $\pi_{12}$ the ``projection'' from $\zG=\zG_0 \times \zG_1 \times \zG_2$ to 
$\zG_1 \times \zG_2$ which associates to $(g_0,g_1,g_2)$ the pair $(g_1,g_2)$
and by $\pi_0$ the projection from $\zG$ to $\zG_0$ which maps 
$(g_0,g_1,g_2)$ to $g_0$.
As in the previous section, we solve with a  quantum walk the following problem:
we fix an element $r=(r_1,r_2)$ in $\zG_1 \times \zG_2$ and
find (if it exists) $(v_{00},v_{01},v_{10},v_{11})$ in $\zV_{00}\times \zV_{01} \times \zV_{10} \times \zV_{11}$ such that 
\begin{eqnarray*}
\pi_{12}(f(v_{00})) +\pi_{12}(f(v_{01})) & = & r \\
\pi_{12}(f(v_{10}))  +\pi_{12}(f(v_{11})) & = & \pi_{12}(S)-r \\
\pi_0(f(v_{00})) + \pi_0(f(v_{01})) + \pi_0(f(v_{10})) + \pi_0(f(v_{11})) & = &\pi_0(S)\\
g(v_{00},v_{01},v_{10},v_{11}) & = & 0.
\end{eqnarray*}
The difference with Proposition \ref{prop:easy} is that we do not check all possibilities for $r$ but just all possibilities for $r_1 \in \zG_1$ and
fix $r_2$ arbitrarily. As in Proposition \ref{prop:easy}, we perform a quantum walk whose complexity is $\OOt{V^{4/5}}$ to 
solve the aforementioned problem for a fixed $r$. 
What remains to be done is to find the right value for $r_1$ which is achieved by a Grover 
search with complexity $\OO{|\zG_1|^{1/2}}$.$~\qed$
\end{proof}
\begin{figure}[h]
    \centering
    \includegraphics[width=0.5\textwidth]{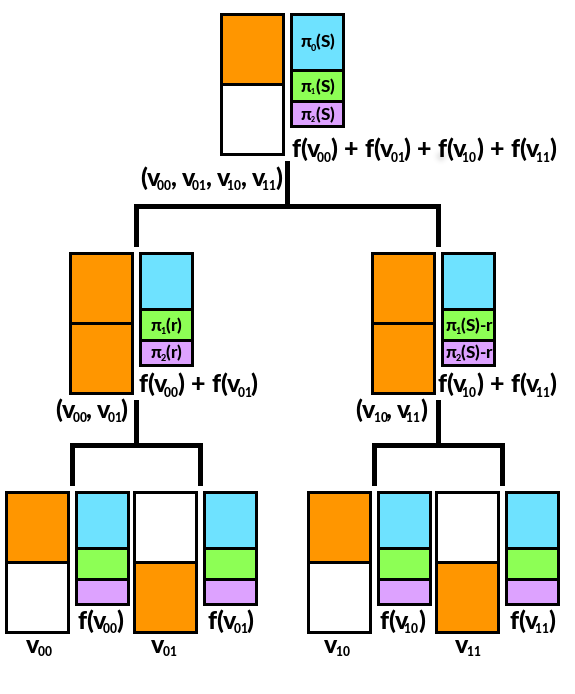}
    \caption{The representation technique: 
the support of the elements of $\zV_{ij}$ is represented in orange, while
the blue, green and violet colours represent $\zG_0$ resp. $\zG_1$, resp. $\zG_2$.}
    \label{fig:mmt}
\end{figure}
By  applying Proposition \ref{prop:improvement} in our decoding context, we obtain
\begin{restatable}{theorem}{thmExpMMTQW}
\label{th:expMMTQWdiese}
We can decode $w=\omega n$ errors in a random linear code of length $n$ and rate $R=\frac{k}{n}$  with a quantum complexity of 
order $ \OOt{2^{\alpha_{\text{MMTQW}}(R,\omega)n}}$ where
\begin{eqnarray*}
    \alpha_{\text{{\tiny{MMTQW}}}}(R,\omega) &\eqdef& \min_{(\pi,\Delta \pi, \lambda) \in \zR} \left( 
\frac{  \beta(R,\lambda,\pi,\Delta \pi)+ \gamma(R,\lambda,\pi,\omega)}{2}\right) \\
&\text{with} & \\
\beta(R,\lambda,\pi,\Delta \pi) & \eqdef & \frac{6}{5} (R+\lambda)H_2\left(\frac{\pi/2+\Delta \pi}{R+\lambda}\right)- \pi - (1-R-\lambda)H_2\left(\frac{\Delta \pi }{1-R-\lambda}\right),\\
\gamma(R,\lambda,\pi,\omega) & \eqdef & H_2(\omega) - (1-R-\lambda)H_2(\frac{\omega - \pi}{1-R-\lambda}) - (R+\lambda)H_2\left(\frac{\pi}{R+\lambda}\right)
\end{eqnarray*}
where $\zR$ is the subset of elements $(\pi,\Delta \pi, \lambda)$ of 
$[0,\omega]  \times [0,1) \times [0,1) $
 that satisfy the following constraints
\begin{eqnarray*}
0 \leq  &\Delta \pi  &\leq R + \lambda - \pi \\
0 \leq &\pi &\leq \min(\omega, R + \lambda)\\
0 \leq &\lambda &\leq 1 - R - \omega + \pi \\
&\pi &= 2\left((R+\lambda)H_2^{-1}\left(\frac{5\lambda}{4(R+\lambda)}\right) - \Delta \pi\right)
\end{eqnarray*}
\end{restatable}
\begin{proof}
The algorithm picks random subsets $\zS$ of size $k+\ell$ with the hope that the restriction to $\zS$ of the error
of weight $w$ that we are looking for is of weight $p$. 
Then it solves for each of these subsets the generalised
$4$-sum problem where the sets $\zV_{ij}$ are specified in \eqref{eq:v00v10} and \eqref{eq:v01v11}, 
and $\zG$, $\zE$, $f$ and $g$ are as  in Problem \ref{pb:shamir_shroeppel}.
$g$ is in this case slightly more complicated for the sake of analysing the algorithm.
We have 
$g(v_{00},v_{01},v_{10},v_{11})=0$ 
 if and only if (i) $v_{00}+v_{01}+v_{10}+v_{11}$ is of weight $p$ 
(this is the additional constraint we use for the analysis of the algorithm)
(ii) $f(v_{00})+f(v_{01})+f(v_{10})+f(v_{11})=\Sigma(s,H,\zS)$ and (iii) 
$h(v_{00}+v_{01}+v_{10}+v_{11})$ is of weight $w$.

From \eqref{eq:T_quant_ISD} we  know that the quantum complexity is given by
\begin{equation}
\label{eq:complexityMMTQW}
\tilde O\left(\frac{T_{\text{MMTQW}}}{\sqrt{P_{\text{MMTQW}}}}\right)
\end{equation}
where $T_{\text{MMTQW}}$ is the complexity of the combination of Grover's algorithm and quantum walk
solving the generalised $4$-sum problem specified above and $P_{\text{MMTQW}}$ is the probability
that the restriction $e'$ of the error $e$ to $\zS$ is of weight $p$
and that this error can be written as $e' = v_{00}+v_{01}+v_{10} + v_{11}$ where the $v_{ij}$ belong to 
$\zV_{ij}$.
It is readily verified that 
$$
P_{\text{MMTQW}} = \OOt{\frac{ \binom{k+\ell}{p} \binom{n-k-\ell}{w - p}}{\binom{n}{w}}}
$$
By using asymptotic expansions of the binomial coefficients we obtain
\begin{equation}
\label{eq:PMMTQW}
(P_{\text{MMTQW}})^{-1/2} = \OOt{2^{\frac{H_2(\omega) - (1-R-\lambda)H_2\left(\frac{\omega - \pi}{1 - R - \lambda}\right) - (R+\lambda)H_2\left(\frac{\pi}{R+\lambda}\right) }{2} n}}
\end{equation}
where $\lambda \eqdef \frac{\ell}{n}$ and $\pi \eqdef \frac{p}{n}$. To estimate $T_{\text{SSQW}}$, we can use Proposition \ref{prop:improvement}.
The point is that the number of different solutions of the generalised $4$-sum problem 
(when there is one) is of order 
$$
\Omt{\binom{p}{p/2} \binom{k+\ell-p}{\Delta p}}.
$$
At this point, we observe that
$$
\log_2\left(\binom{p}{p/2} \binom{k+\ell-p}{\Delta p}\right) = p+ (k+\ell-p)H_2\left( \frac{\Delta p}{k+\ell -p} \right) +o(n)
$$
when $p$, $\Delta p$, $\ell$, $k$ are all linear in $n$.
In other words, we may use Proposition $\ref{prop:improvement}$ with 
$\zG_2 = \Ft^{\ell_2}$ with 
\begin{equation}
\label{eq:ell2}
\ell_2 \eqdef p+ (k+\ell-p)H_2\left( \frac{\Delta p}{k+\ell -p} \right).
\end{equation}
We use now Proposition \ref{prop:improvement}
with $\zG_2$ chosen as explained above. 
$V$ is given in this case by
$$
V = \binom{\frac{k+\ell}{2}}{\frac{p}{4}+\frac{\Delta p}{2}} = \OOt{2^{\frac{(R+\lambda) H_2\left(\frac{\pi/2 + \Delta \pi}{R + \lambda}\right)n}{2}}}
$$
where $\Delta \pi \eqdef \frac{\Delta p}{n}$.
We choose the size of $\zG$ such that
\begin{equation}\label{eq:G}
|\zG| = \Tht{V^{8/5}}
\end{equation}
which gives
$$
2^\ell = \Tht{  {\binom{\frac{k+\ell}{2}}{\frac{p}{4}+\frac{\Delta p}{2}}}^{8/5}}.
$$
This explains why we impose 
$$
\lambda = \frac{8}{5}\frac{R+\lambda}{2} H_2\left(\frac{\pi/2+\Delta \pi}{R+\lambda}\right)
$$
which is equivalent to the condition 
$$
\frac{5 \lambda}{4(R+\lambda)} = H_2\left(\frac{\pi/2+\Delta \pi}{R+\lambda}\right)
$$
which in turn is equivalent to the condition
\begin{equation}\label{eq:condition2MMT}
\pi = 2\left((R+\lambda)H_2^{-1}\left(\frac{5\lambda}{4(R+\lambda)}\right) - \Delta \pi\right)
\end{equation}
found in the definition of the region $\zR$.
The size of $\zG_1$ is chosen such that
\begin{equation}\label{eq:condition1MMT}
|\zG_1| \cdot |\zG_2| = \Ft^{\lceil \frac{\ell}{2} \rceil}.
\end{equation}
By using \eqref{eq:ell2} and \eqref{eq:G}, this implies
\begin{eqnarray}
|\zG_1| & = & \Tht{\frac{V^{4/5}}{2^{p+ (k+\ell-p)H_2\left( \frac{\Delta p}{k+\ell -p} \right)}}}
\end{eqnarray}

With the choices \eqref{eq:condition1MMT} and \eqref{eq:condition2MMT}, 
we obtain
\begin{eqnarray}
T_{\text{MMTQW}} &= &\OOt{|\zG_1|^{1/2} \cdot V^{4/5}} \nonumber \\
& = & \OOt{\frac{V^{6/5}}{2^{\frac{p}{2}+ \frac{k+\ell-p}{2} H_2\left( \frac{\Delta p}{k+\ell -p} \right)}}} \nonumber \\
& = & \OOt{ 2^{\left[\frac{3}{5}(R+\lambda)H_2\left( \frac{\pi/2 +\Delta \pi}{R+\lambda} \right)- \frac{\pi}{2} - \frac{R+\lambda-\pi}{2} H_2\left( \frac{\Delta \pi}{R+\lambda -
\pi} \right) \right]n}}\label{eq:TMMTQW}
\end{eqnarray}
Substituting for $P_{\text{MMTQW}}$ and $T_{\text{MMTQW}}$ the expressions given by \eqref{eq:PMMTQW} and \eqref{eq:TMMTQW} finishes the proof of the
theorem.$~\qed$
\end{proof}

\section{Computing the complexity exponents	}
We used the software SageMath to numerically find the minima giving the complexity exponents in Theorems
\ref{th:expSSQW} and \ref{th:expMMTQWdiese} using golden section search and a recursive version thereof for two parameters.
We compare in Figure \ref{fig:complexities} the exponents $\alpha_{\text{Bernstein}}(R,\omega_{\text{GV}})$, $\alpha_{SSQW}(R,\omega_{\text{GV}})$ 
and $\alpha_{MMTQW}(R,\omega_{\text{GV}})$ that we have obtained with our approach. It can be observed that there is some improvement upon 
$\alpha_{\text{Bernstein}}$ with both algorithms especially in the range of rates between 0.3 and 0.7.
\begin{figure}[h!]
    \centering
    \includegraphics[width=0.6\textwidth]{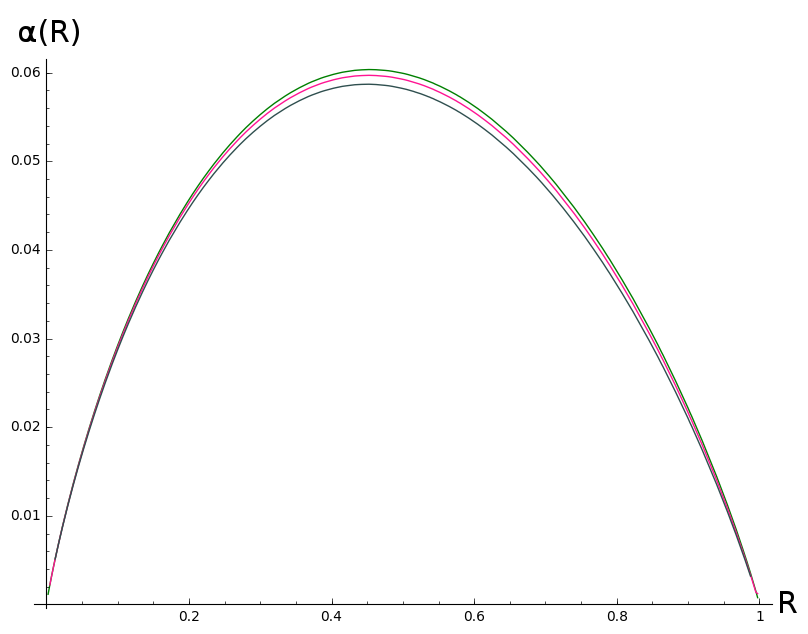}
    \caption{$\alpha_{\text{Bernstein}}$ in green, $\alpha_{SSQW}$ in pink, $\alpha_{MMTQW}$ in grey. \label{fig:complexities}}
    \label{fig:courbes}
\end{figure}
\section{Concluding remarks}
One may wonder why our best algorithm is a version of MMT's algorithm  and not BJMM's algorithm or May and Ozerov's algorithm. We did try to quantise BJMM's algorithm, but it turned out to have worse time complexity than MMT's algorithm (for more details, see \cite{K16}). 
This seems to be due to space complexity constraints. Space complexity is indeed a lower bound on the quantum complexity of the algorithm. It has actually   been shown \cite[Chap. 10, Sec. 3]{B12} that BJMM's algorithm uses more space than MMT's algorithm, even when it is optimised to use the least amount of space. Moreover, it is rather insightful that in all cases, the best quantum algorithms that we have obtained here are not direct quantised versions of the original Dumer or MMT's algorithms but quantised versions of modified versions of these algorithms that use less memory than the original
algorithms.

The case of the May and Ozerov algorithm is also intriguing. Again the large space complexity of the original version of this algorithm makes it a very challenging task to obtain a ``good'' quantised version of it.

Finally, it should be noticed that while sophisticated techniques such as MMT, BJMM \cite{MMT11,BJMM12} or May and Ozerov \cite{MO15} have managed to improve rather significantly upon the most naive ISD algorithm, namely Prange's algorithm \cite{P62}, the improvement that we obtain with more sophisticated techniques is much more modest when we consider our improvements of the quantised version of the Prange algorithm \cite{B10}. Moreover, the improvements we obtain on the exponent $\alpha_{\text{Bernstein}}(R,\omega)$ are smaller when $\omega$ is smaller than $\omega_{\text{GV}}$. Considering the techniques for proving that the exponent of classical ISD algorithms goes to the Prange exponent when the 
relative error weight goes to $0$ \cite{CS16}, we conjecture that it should be possible to prove that we actually have $\lim_{\omega \rightarrow 0^+} \frac{\alpha_{\text{MMTQW}}(R,\omega)}{\alpha_{\text{Bernstein}}(R,\omega)} = 1$.

\bibliographystyle{acm}
\addcontentsline{toc}{section}{Bibliography}
\bibliography{codecrypto}

\begin{thebibliography}{10}

\bibitem{A07}
{\sc Ambainis, A.}
\newblock Quantum walk algorithm for element distinctness.
\newblock {\em SIAM J. Comput. 37\/} (2007), 210--239.

\bibitem{B97b}
{\sc Barg, A.}
\newblock Complexity issues in coding theory.
\newblock {\em Electronic Colloquium on Computational Complexity\/} (Oct.
  1997).

\bibitem{B12}
{\sc Becker, A.}
\newblock {\em The representation technique, Applications to hard problems in
  cryptography}.
\newblock PhD thesis, Universit{\'e} Versailles Saint-Quentin en Yvelines, Oct.
  2012.

\bibitem{BJMM12}
{\sc Becker, A., Joux, A., May, A., and Meurer, A.}
\newblock Decoding random binary linear codes in {$2^{n/20}$}: How {$1+1=0$}
  improves information set decoding.
\newblock In {\em Advances in Cryptology - EUROCRYPT~2012\/} (2012), Lecture
  Notes in Comput. Sci., Springer.

\bibitem{B10}
{\sc Bernstein, D.~J.}
\newblock Grover vs. {McEliece}.
\newblock In {\em Post-Quantum Cryptography~2010\/} (2010), N.~Sendrier, Ed.,
  vol.~6061 of {\em Lecture Notes in Comput. Sci.}, Springer, pp.~73--80.

\bibitem{BJLM13}
{\sc Bernstein, D.~J., Jeffery, S., Lange, T., and Meurer, A.}
\newblock Quantum algorithms for the subset-sum problem.
\newblock In {\em Post-Quantum Cryptography~2011\/} (Limoges, France, June
  2013), vol.~7932 of {\em Lecture Notes in Comput. Sci.}, pp.~16--33.

\bibitem{BLP11}
{\sc Bernstein, D.~J., Lange, T., and Peters, C.}
\newblock Smaller decoding exponents: ball-collision decoding.
\newblock In {\em Advances in Cryptology - CRYPTO~2011\/} (2011), vol.~6841 of
  {\em Lecture Notes in Comput. Sci.}, pp.~743--760.

\bibitem{BBHT98}
{\sc Boyer, M., Brassard, G., H{\o}yer, P., and Tapp, A.}
\newblock Tight bounds on quantum searching.
\newblock {\em Fortsch. Phys. 46\/} (1998), 493.

\bibitem{CS16}
{\sc Canto-Torres, R., and Sendrier, N.}
\newblock Analysis of information set decoding for a sub-linear error weight.
\newblock In {\em Post-Quantum Cryptography~2016\/} (Fukuoka, Japan, Feb.
  2016), Lecture Notes in Comput. Sci., pp.~144--161.

\bibitem{CDS80}
{\sc Cvetkovi{\'c}, D.~M., Doob, M., and Sachs, H.}
\newblock {\em {Spectra of graphs : theory and application}}.
\newblock New York : Academic Press, 1980.

\bibitem{D91}
{\sc Dumer, I.}
\newblock On minimum distance decoding of linear codes.
\newblock In {\em Proc. 5th Joint Soviet-Swedish Int. Workshop Inform.
  Theory\/} (Moscow, 1991), pp.~50--52.

\bibitem{FS09}
{\sc Finiasz, M., and Sendrier, N.}
\newblock Security bounds for the design of code-based cryptosystems.
\newblock In {\em Advances in Cryptology - ASIACRYPT~2009\/} (2009), M.~Matsui,
  Ed., vol.~5912 of {\em Lecture Notes in Comput. Sci.}, Springer, pp.~88--105.

\bibitem{G96a}
{\sc Grover, L.~K.}
\newblock A fast quantum mechanical algorithm for database search.
\newblock In {\em Proc. 28th Annual ACM Symposium on the Theory of
  Computation\/} (New York, NY, 1996), ACM Press, New York, pp.~212--219.

\bibitem{G97}
{\sc Grover, L.~K.}
\newblock Quantum computers can search arbitrarily large databases by a single
  query.
\newblock {\em Phys. Rev. Lett. 79\/} (1997), 4709--4712.

\bibitem{HJ10}
{\sc {Howgrave-Graham}, N., and Joux, A.}
\newblock New generic algorithms for hard knapsacks.
\newblock In {\em Advances in Cryptology - EUROCRYPT~2010\/} (2010),
  H.~Gilbert, Ed., vol.~6110 of {\em Lecture Notes in Comput. Sci.}, Sringer.

\bibitem{K16}
{\sc Kachigar, G.}
\newblock {\'E}tude et conception d'algorithmes quantiques pour le d{\'e}codage
  de codes lin{\'e}aires.
\newblock Master's thesis, {Universit{\'e} de Rennes 1, France}, Sept. 2016.

\bibitem{MNRS07}
{\sc Magniez, F., Nayak, A., Roland, J., and Santha, M.}
\newblock Search via quantum walk.
\newblock In {\em Proceedings of the Thirty-ninth Annual ACM Symposium on
  Theory of Computing\/} (2007), STOC '07, pp.~575--584.

\bibitem{MMT11}
{\sc May, A., Meurer, A., and Thomae, E.}
\newblock Decoding random linear codes in {$O(2^{0.054n})$}.
\newblock In {\em Advances in Cryptology - ASIACRYPT~2011\/} (2011), D.~H. Lee
  and X.~Wang, Eds., vol.~7073 of {\em Lecture Notes in Comput. Sci.},
  Springer, pp.~107--124.

\bibitem{MO15}
{\sc May, A., and Ozerov, I.}
\newblock On computing nearest neighbors with applications to decoding of
  binary linear codes.
\newblock In {\em Advances in Cryptology - EUROCRYPT~2015\/} (2015), E.~Oswald
  and M.~Fischlin, Eds., vol.~9056 of {\em Lecture Notes in Comput. Sci.},
  Springer, pp.~203--228.

\bibitem{M78}
{\sc McEliece, R.~J.}
\newblock {\em A Public-Key System Based on Algebraic Coding Theory}.
\newblock Jet Propulsion Lab, 1978, pp.~114--116.
\newblock DSN Progress Report 44.

\bibitem{N86}
{\sc Niederreiter, H.}
\newblock Knapsack-type cryptosystems and algebraic coding theory.
\newblock {\em Problems of Control and Information Theory 15}, 2 (1986),
  159--166.

\bibitem{OS09}
{\sc Overbeck, R., and Sendrier, N.}
\newblock Code-based cryptography.
\newblock In {\em Post-quantum cryptography\/} (2009), D.~J. Bernstein,
  J.~Buchmann, and E.~Dahmen, Eds., Springer, pp.~95--145.

\bibitem{P62}
{\sc Prange, E.}
\newblock The use of information sets in decoding cyclic codes.
\newblock {\em {IRE} Transactions on Information Theory 8}, 5 (1962), 5--9.

\bibitem{S08}
{\sc Santha, M.}
\newblock {Quantum walk based search algorithms}.
\newblock In {\em {5th TAMC}\/} (2008), pp.~31--46.
\newblock arXiv/0808.0059.

\bibitem{SS81}
{\sc Schroeppel, R., and Shamir, A.}
\newblock A {$T=O(2^{n/2})$, $S=O(2^{n/4})$} algorithm for certain
  {NP}-complete problems.
\newblock {\em SIAM J. Comput. 10}, 3 (1981), 456--464.

\bibitem{S97}
{\sc Shor, P.~W.}
\newblock Polynomial-time algorithms for prime factorization and discrete
  logarithms on a quantum computer.
\newblock {\em SIAM J. Comput. 26}, 5 (1997), 1484--1509.

\bibitem{S88}
{\sc Stern, J.}
\newblock A method for finding codewords of small weight.
\newblock In {\em Coding Theory and Applications\/} (1988), G.~D. Cohen and
  J.~Wolfmann, Eds., vol.~388 of {\em Lecture Notes in Comput. Sci.}, Springer,
  pp.~106--113.

\bibitem{S04}
{\sc Szegedy, M.}
\newblock Quantum speed-up of markov chain based algorithms.
\newblock In {\em Proc. of the 45th IEEE Symposium on Foundations of Computer
  Science\/} (2004), pp.~32--41.

\end{thebibliography}
\newpage
\appendix
\section{Proofs for Section 2}
We want to prove the following theorem.
\thmProdJohnsonGraph*
We need the following results for the proof.
\begin{theorem}[Cartesian product of $d$-regular graphs]
\label{thm:productgraphs_regular}
Let $n \in \mathbb{N}$ and $G_1, ..., G_n$ be undirected $d$-regular graphs. Then $G^n = \times_{i=1}^n G_i$ has $\prod_{i=1}^n |V_i|$ vertices and is $nd$-regular.
\end{theorem}
The proof of this theorem is immediate.
\begin{theorem}[Spectral gap of product graphs]
\label{thm:spectralgap_productgraphs}
Let $G_1$ and $G_2$ be $d_1$- resp. $d_2$-regular graphs with eigenvalues of the associated Markov chain $1 = \lambda_1 \geq~...~\geq \lambda_{k_1}$ resp. $1 = \mu_1 \geq~...~ \geq \mu_{k_2}$.
Denote by $\delta_i$ the spectral gap of $G_i$, $i=1,2$.\\
Then the spectral gap $\delta$ of the product graph $G_1 \times G_2$ fulfils:
$$\delta \geq \frac{\min\left(\delta_2d_2, \delta_1d_1\right)}{d_1 + d_2}$$
\end{theorem}
\begin{proof}
We first recall the following result (see \cite{CDS80}, Chapter 2, Section 5, Theorems 2.23 and 2.24):\\
The Markov chain associated to the graph $G_1 \times G_2$ has $k_1k_2$ eigenvalues which are $\nu_{i,j} = \frac{d_1\lambda_i + d_2\mu_j}{d_1 + d_2}$. In particular, $\delta = \frac{d_1 + d_2 - \max_{(i,j) \not = (1,1)}|d_1\lambda_i  + d_2\mu_j|}{d_1 + d_2}$.\\
As the eigenvalues of $G_1$ and $G_2$ are ordered from largest to smallest, we have the following:
$$\max_{i=2,...,k_1}|\lambda_i| = \max(\lambda_2,-\lambda_{k_1})$$
$$\max_{j=2,...,k_2}|\mu_i| = \max(\mu_2,-\mu_{k_2})$$
Furthermore
$$d_1\delta_1 = d_1 - d_1\max_{i=2,...,k_1}|\lambda_i| = d_1 - d_1\max(\lambda_2,-\lambda_{k_1}) \leq d_1 + d_1\lambda_{k_1}$$
$$d_2\delta_2 = d_2 - d_2\max_{j=2,...,k_2}|\mu_j| = d_2 - d_2\max(\mu_2,-\mu_{k_2}) \leq d_2 + d_2\mu_{k_2}$$
Which taken together entail\\
$$d_1 + d_2 + d_1\lambda_{k_1} + d_2\mu_{k_2} \geq d_1\delta_1 + d_2\delta_2$$
Moreover
\begin{align*}
  \max_{(i,j) \not = (1,1)}|d_1\lambda_i  + d_2\mu_j| &= \max\left(d_1\lambda_1 + d_2\mu_2, d_1\lambda_2 + d_2\mu_1, - d_1\lambda_{k_1} - d_2\mu_{k_2}\right)\\ 
  &= \max\left(d_1 + d_2\mu_2, d_1\lambda_2 + d_2, - d_1\lambda_{k_1} - d_2\mu_{k_2}\right)\\
\end{align*}
Therefore
$$(d_1 + d_2)\delta = d_1 + d_2 - \max_{(i,j) \not = (1,1)}|d_1\lambda_i  + d_2\mu_j| = \min\left(d_2 - d_2\mu_2, d_1 - d_1\lambda_2, d_1 + d_2 + d_1\lambda_{k_1} + d_2\mu_{k_2}\right)$$\\
Finally
\begin{align*}
  \delta &= \frac{d_1 + d_2 - \max_{(i,j) \not = (1,1)}|d_1\lambda_i  + d_2\mu_j|}{d_1 + d_2}\\
  &= \frac{\min\left(d_2 - d_2\mu_2, d_1 - d_1\lambda_2, d_1 + d_2 + d_1\lambda_{k_1} + d_2\mu_{k_2}\right)}{d_1 + d_2}\\
  &\geq \frac{\min\left(d_2 - d_2\max(\mu_2,-\mu_{k_2}), d_1 - d_1\max(\lambda_2,-\lambda_{k_1}), d_1\delta_1 + d_2\delta_2\right)}{d_1 + d_2}\\
  &\geq \frac{\min\left(\delta_2d_2, \delta_1d_1, d_1\delta_1 + d_2\delta_2\right)}{d_1 + d_2}\\
  &\geq \frac{\min\left(\delta_2d_2, \delta_1d_1\right)}{d_1 + d_2}
\end{align*}
$\qed$
\end{proof}
\begin{proof}[Theorem \ref{thm:product_johnsongraphs}]
  Point (1) is immediate by Theorem \ref{thm:productgraphs_regular}.\\
  Point (2) is proved using induction.\\
  Indeed, for $m=2$ we have, by Theorem \ref{thm:spectralgap_productgraphs} :
    $$\delta(J^2) \geq \frac{\delta(J)d}{2d} = \frac{1}{2}\delta(J)$$
  And for $m\geq2$, supposing that $\delta(J^m) \geq \frac{1}{m}\delta(J)$, we have, using Theorem \ref{thm:spectralgap_productgraphs} and point (1) of this theorem :
    \begin{align*}
      \delta(J^{m+1}) &\geq \frac{\min\left(md\delta(J^m),d\delta(J)\right)}{md + d}\\
      &\geq \frac{\min\left(\frac{md}{m}\delta(J),d\delta(J)\right)}{md + d}\\
      &\geq \frac{1}{m+1}\delta(J)
    \end{align*}
Point (3) just follows from the fact that $J^m(n,r)$ is regular, undirected, connected and has positive spectral gap by using the previous point.
    $\qed$	
\end{proof}

\end{document}